\def\LMtwocolumn{1}
\documentclass[journal]{IEEEtran}
\usepackage[T1]{fontenc}
\usepackage{amsmath,amssymb,amsthm,cite,graphicx,tikz}
\usepackage{xurl}
\usepackage[hidelinks]{hyperref}
\newtheorem{theorem}{Theorem}
\newtheorem{externaltheorem}{Theorem}[section]

\newtheorem{lemma}{Lemma}

\newtheorem{corollary}{Corollary}
\theoremstyle{definition}
\newtheorem{definition}{Definition}

\allowdisplaybreaks[1]
\newcommand{\LMfigtext}[2]{#1}

\title{When Multiple-Unicast Feasibility Is Equivalent to Single-Sink Routability}
\author{Sirui~Liu,~\IEEEmembership{Student Member,~IEEE,}
        Zongpeng~Li,~\IEEEmembership{Senior Member,~IEEE,}
        and~Baochun~Li,~\IEEEmembership{Fellow,~IEEE}%
\thanks{Corresponding author: Zongpeng Li.}
\thanks{Sirui Liu is with the Institute for Network Sciences \& Cyberspace, Tsinghua University, Beijing, China, and also with Quan Cheng Laboratory, Shandong, China (email: liu-sr24@mails.tsinghua.edu.cn).}
\thanks{Zongpeng Li is with Quan Cheng Laboratory, Shandong, China, and also with the Institute for Network Sciences \& Cyberspace, Tsinghua University, Beijing, China (email: zongpeng@tsinghua.edu.cn).}
\thanks{Baochun Li is with the Department of Electrical and Computer Engineering, University of Toronto, Toronto, ON, Canada (email: bli@ece.toronto.edu).}}

\date{}
\hypersetup{pdftitle={When Multiple-Unicast Feasibility Is Equivalent to Single-Sink Routability},pdfauthor={Sirui Liu, Zongpeng Li, Baochun Li}}
\begin{document}
\maketitle
\begin{abstract}
In communication networks, network coding allows intermediate nodes to combine messages, achieving transmission rates unattainable by routing alone in many multicast settings. In multiple unicast, where each source sends a message only to its prescribed destination, a central question is under what conditions fractional routing suffices to attain the transmission rates supported by network coding. We prove that, for any given collection of multiple-unicast demands in an undirected network, if the maximum flow between every pair of distinct destinations is at least the total demand rate, then the following are equivalent in feasibility: multiple-unicast coding, multiple-unicast fractional routing, and simultaneous fractional routing of all source rates to any fixed destination. Verifying this pairwise maximum-flow condition and, when it holds, deciding feasibility of the given multiple-unicast demands require at most as many maximum-flow computations as there are destinations. As a direct corollary, we prove the Langberg--M\'edard conjecture: if a directed network supports coded multicast of independent source messages at a common rate to every destination, then removing arc directions without changing capacities permits simultaneous fractional routing at that rate for all prescribed source--destination pairs.
\end{abstract}

\section{Introduction}\label{sec:introduction}

Network coding allows intermediate nodes to combine received messages~\cite{ahlswede2000}, potentially achieving transmission rates higher than those attainable by routing. In contrast, fractional routing only forwards messages, while allowing each source's traffic to be split among multiple paths. A basic question in network information theory is therefore under what conditions fractional routing suffices to attain the transmission rates supported by network coding.

In multiple-unicast communication, each source sends an independent message to its prescribed destination. The multiple-unicast conjecture proposed by Li and Li~\cite{lili2004} in 2004 asserted that fractional routing always attains the rates supported by network coding in undirected networks, where each edge can carry traffic in either direction and the two directions share its capacity. However, Zhang, Li, and Li~\cite{zhang2026} and Braverman and He~\cite{braverman2026} recently gave explicit counterexamples, showing that even in undirected networks, coding can support rates unattainable by fractional routing. Earlier results of OpenAI on integer multiplication and exact Fourier circuits~\cite{openai2026integer,openai2026fourier} also imply counterexamples. Li, Li, and Liu~\cite{liliLiu2026sorting} further used coding gains to accelerate external-memory integer sorting. Thus, allowing the two directions of a link to share its capacity is not by itself sufficient to ensure that fractional routing attains coding rates.

In 2009, Langberg and M\'edard~\cite{langberg2009} proposed a weaker version of the multiple-unicast conjecture that assumes multiple-multicast coding in an initial directed network. Suppose a directed network with $k$ source--destination pairs supports coded multiple multicast: the sources send independent messages at a common rate, and every destination recovers all source messages, including its own prescribed message. The Langberg--M\'edard conjecture asserts that removing all arc directions without increasing capacities yields an undirected network in which fractional routing can simultaneously deliver the prescribed messages to their respective destinations at the original common source rate.

For unit arc capacities and common rate one, the feasibility of this coded multiple-multicast task has a path characterization: for each destination considered separately, there are $k$ arc-disjoint directed paths connecting the $k$ sources, respectively, to that destination, a property called strong reachability~\cite{langberg2009}. In the underlying undirected network, the target routing must simultaneously carry one unit of flow for every prescribed source--destination pair, with each edge's unit capacity shared by both directions; the difficulty is that the directed path families for different destinations may overlap, so combining these paths directly may create capacity conflicts. Prior work on strongly reachable networks established common routing rates of $1/3$~\cite{langberg2009} and $8/9$~\cite{cai2017,cai2023} for any number of pairs, and proved unit-rate routability for three-pair directed acyclic networks~\cite{cai2020,cai2022} and networks whose collapsed backbone is a rooted binary tree~\cite{cai2024}.

Rather than attempting to untangle overlapping directed path families directly, we examine the flow properties that multiple-multicast coding imparts to the underlying undirected network: after removing arc directions, all sources can simultaneously send their traffic at the original rates to any chosen destination, forming a single-sink routing, and the maximum flow between any two distinct destinations is at least the total source rate.

Motivated by this observation, we prove a feasibility equivalence theorem for undirected networks (Theorem~\ref{thm:main-equivalence}): for any collection of multiple-unicast demands with given sources, destinations, and rates, if the maximum flow between every pair of distinct destinations is at least the total demand rate, then simultaneous fractional routing to the prescribed destinations, network coding for the same demands, and simultaneous fractional routing of all source rates to a single fixed destination are equivalent in feasibility. Intuitively, single-sink routability ensures that the aggregate traffic from all sources can enter the destination set without bottlenecks, while pairwise maximum flow of at least the total demand rate ensures that this traffic can be delivered simultaneously to the respective targets through the network. Because the common destination may be fixed in advance, verifying this maximum-flow condition and, when it holds, deciding multiple-unicast coding or routing feasibility require at most as many maximum-flow computations as there are destinations (Corollary~\ref{cor:maxflow-test}). For a directed network supporting multiple-multicast coding, applying Theorem~\ref{thm:main-equivalence} to its underlying undirected network immediately proves the Langberg--M\'edard conjecture (Corollary~\ref{cor:lm-conjecture}).

To prove Theorem~\ref{thm:main-equivalence}, we use linear programming duality for multicommodity flow~\cite{iri1971,onaga1971}, which compares the total capacity cost under any nonnegative edge pricing with the rate-weighted sum of shortest-path costs for the demands. First, using a theorem of Kir\'aly~\cite[Theorem~2.1]{kiraly2007}, we eliminate internal nodes and edges between sources while preserving single-sink routability and the pairwise maximum-flow condition between destinations, without increasing the total capacity cost. Next, Mader's splitting-off theorem~\cite{mader1978},~\cite[Theorem~1.2]{lau2013} lets us eliminate each source by replacing pairs of its incident edges with edges between their destination endpoints, preserving the maximum flows between destinations. An inequality involving distances among three destinations then gives the required routing-cost bound after the added-edge costs are subtracted from both sides. This method of successively eliminating internal nodes and sources also applies to the analysis of other multicommodity-flow and network information-flow problems.

The remainder of the paper is organized as follows. Section~\ref{sec:related_work} reviews related work. Section~\ref{sec:model} introduces the network model, fractional-routing duality, and the flow characterization of multiple-multicast coding. Section~\ref{sec:results} states the main results and outlines the proof. Section~\ref{sec:core} proves the metric inequality and routing theorem when sources are adjacent only to destinations. Section~\ref{sec:reduction} eliminates internal nodes and edges between sources to complete the proof of Theorem~\ref{thm:main-equivalence}. Section~\ref{sec:strong_reachability} deduces the Langberg--M\'edard conjecture from Theorem~\ref{thm:main-equivalence}, and Section~\ref{sec:conclusion} concludes. Appendix~\ref{sec:external-results} states the forms of Mader's and Kir\'aly's theorems used in the proofs.

\section{Related Work}\label{sec:related_work}

\subsection{The Multiple-Unicast Conjecture}
Li and Li~\cite{lili2004} and, independently, Harvey, Kleinberg, and Lehman~\cite{harvey2004} proposed the multiple-unicast conjecture in 2004. The conjecture asserts that network coding and fractional routing have the same capacity region in undirected multiple-unicast networks. One line of positive results follows from cut sufficiency on special graph classes: every rate vector achievable by coding satisfies the cut condition, so the conjecture holds whenever this condition also guarantees fractional routing. Representative network classes follow from Hu's two-commodity theorem~\cite{hu1963}, Seymour's four-terminal theorem~\cite{seymour1980}, and the planar-network results of Okamura and Seymour~\cite{okamura1981} and Okamura~\cite{okamura1983}.

When the cut condition does not characterize routing feasibility, message independence and decodability provide further constraints on coding rates. Jain et al.~\cite{jain2006} introduced the Input--Output and Crypto inequalities and proved equality of coding and routing rates for a four-session instance on $K_{3,2}$. Harvey et al.~\cite{harvey2006} introduced informational dominance and combined it with entropy inequalities to establish the conjecture for a class of special bipartite networks. Qureshi and Thakor~\cite{qureshi2022} developed the partition bound and showed that routing attains this upper bound on the common coding rate in two classes of complete multipartite networks.

Structural reductions provide another route to extending these results. Yin et al.~\cite{yin2018} developed a reduction method based on relations between cut sets and source--destination paths, established the conjecture for a new network class, and combined the method with computer-aided search to prove the conjecture for networks with at most six coding nodes. Liu et al.~\cite{liu2026session} introduced session dominance, which compares collections of sessions through shortest-path containment. A conclusion for a dominating collection transfers to a dominated collection, yielding reductions for three-session instances and a network class characterized by topological conditions.

Geometric and metric methods study the same question through transmission costs. Li and Wu~\cite{liwu2012} introduced space information flow to compare coding and routing costs in geometric spaces; Xiahou et al.~\cite{xiahou2014} used metric embeddings to connect graph instances with space instances. Liu and Li~\cite{liuli2025} proved equality of coding and routing costs in $\ell_p^n$ for $1 \le p \le 2$. Subsequently, Liu et al.~\cite{liu2026metric} developed a unified graph-metric framework, proving the conjecture for networks with at most five terminal locations, planar networks whose terminals lie on at most three designated face boundaries and whose individual sessions have both endpoints on a common designated face boundary, and a network class specified by the structure of session endpoints.

Recently, Zhang, Li, and Li~\cite{zhang2026} and Braverman and He~\cite{braverman2026} independently gave explicit counterexamples, showing that the conjecture fails in general undirected networks. Together with the positive results above, these counterexamples show that the relationship between coding and routing depends on more specific network conditions. This paper establishes a pairwise maximum-flow condition on the destinations under which multiple-unicast routing, multiple-unicast coding, and single-sink routing are equivalent in feasibility, and uses it to resolve the Langberg--M\'edard conjecture.

\subsection{The Langberg--M\'edard Conjecture}
Langberg and M\'edard~\cite{langberg2009} proposed a weaker version of the multiple-unicast conjecture: if coding in a directed network allows every destination to receive the independent messages sent by all sources at a common rate, then removing all arc directions without increasing capacities should allow simultaneous fractional routing of the prescribed source--destination pairs at that rate. With unit arc capacities and unit common rate, this coding hypothesis is equivalent to strong reachability: for each destination separately, a family of arc-disjoint paths connects all sources to that destination.

For an arbitrary number of source--destination pairs, Langberg and M\'edard~\cite{langberg2009} first established a common routing rate of $1/3$. Cai and Han~\cite{cai2017} improved this guarantee to $8/9$ by combining elementary flows through matrix optimization. Their subsequent work~\cite{cai2023} characterized the asymptotic optimum of that optimization framework and further analyzed its finite-parameter optima.

For unit-rate routing, Cai and Han~\cite{cai2020,cai2022} worked with directed acyclic networks, first proving the conjecture for stable three-pair networks and then for all strongly reachable three-pair networks in that model. They also established the conjecture for any number of pairs when the collapsed backbone has the structure of a rooted binary tree~\cite{cai2024}. We prove the conjecture as a consequence of this feasibility equivalence.

\subsection{Network Reduction}
Splitting off is a basic operation for simplifying a network: two edges incident to the same node are replaced by an edge joining their other endpoints. Mader~\cite{mader1978} proved that, under suitable conditions, this operation can be repeated to eliminate a node while preserving edge connectivities among the remaining nodes. We use the formulation stated by Lau and Yung~\cite[Theorem~1.2]{lau2013} to split off sources that are adjacent only to destinations and use the preserved connectivity to establish a lower bound on routing cost.

Reducing a general network to one containing only sources and destinations requires eliminating internal nodes while preserving single-sink routability and the pairwise maximum-flow condition and controlling cost. Kir\'aly's result on minimum-cost cut covering~\cite[Theorem~2.1]{kiraly2007} provides a tool for this purpose. We use this result to eliminate all internal nodes and edges between sources without increasing cost, thereby reducing the proof to the network structure studied in Section~\ref{sec:core}.

\section{Notation and Preliminaries}\label{sec:model}

\subsection{Network Model}
Let $D = (V, \mathcal{A})$ be a finite directed multigraph, where $V$ is the node set and $\mathcal{A}$ is the arc set. Directed cycles are allowed. An arc $a \in \mathcal{A}$ from $u$ to $v$ permits transmission from $u$ to $v$, and its finite nonnegative capacity $c_a \ge 0$ specifies the maximum transmission rate on that arc. The directed network has sources $s_1, \dots, s_k$ and destinations $t_1, \dots, t_k$, where $k \ge 1$. Source nodes may coincide, destination nodes may coincide, and a source and a destination may be the same node. Replacing every arc $a$ from $u$ to $v$ by an undirected edge $e$ between the same nodes, with capacity $c_e = c_a$, gives the underlying undirected multigraph of $D$, denoted by $\overline{D}$. In $\overline{D}$, each edge permits transmission in either direction, and the sum of the two directional rates cannot exceed $c_e$. Oppositely directed arcs become two parallel undirected edges, whose capacities may be added.

For general undirected routing, let $G = (V, E)$ be a finite undirected multigraph. Each edge $e \in E$ has a finite nonnegative real capacity $c_e \ge 0$ shared by its two directions, and $c = (c_e)_{e \in E}$ denotes the capacity vector. We denote this network by $(G, c)$. Parallel edges may be combined by adding their capacities, and loops may be deleted. For any $U \subseteq V$, the cut $\delta_G U$ consists of the edges with exactly one endpoint in $U$, and its capacity is $c(\delta_G U) = \sum_{e \in \delta_G U} c_e$. For any two distinct nodes $u, v \in V$, write $\operatorname{maxflow}_G(u, v)$ for the single-commodity maximum flow between $u$ and $v$ in $(G, c)$, evaluated separately for each pair. By the max-flow min-cut theorem, $\operatorname{maxflow}_G(u, v)$ equals the minimum of $c(\delta_G U)$ over all cuts $U \subseteq V$ separating $u$ and $v$.

\subsection{Fractional Routing}
For a positive integer $q$, write $[q] = \{1, \dots, q\}$. In the undirected network $(G, c)$, consider a collection of multiple-unicast demands $\mathcal{D} = ((s_i, t_i, r_i))_{i \in [k]}$, where source $s_i \in V$ must send an independent message at rate $r_i \ge 0$ to its prescribed destination $t_i \in V$, and let $R = \sum_{i=1}^k r_i$ be the total demand rate. Messages from different sessions are not encoded together, and source or destination nodes may coincide. Let $\mathcal{P}_i$ be the finite set of all simple paths in $G$ joining $s_i$ to $t_i$; when $s_i = t_i$, $\mathcal{P}_i$ consists of a single path of length zero using no edges. A simultaneous fractional routing of these demands is a family of nonnegative path flows $f_{iP} \ge 0$ satisfying
\begin{equation}
\sum_{P \in \mathcal{P}_i} f_{iP} = r_i \quad (i \in [k]), \label{eq:flow-demand}
\end{equation}
\begin{equation}
\sum_{i=1}^k \sum_{\substack{P \in \mathcal{P}_i \\ e \in P}} f_{iP} \le c_e \quad (e \in E). \label{eq:flow-capacity}
\end{equation}
Thus, each source $s_i$ may split $r_i$ units of traffic among several paths, and the total flow from all sources on an edge $e$, summed over both directions, cannot exceed $c_e$. Fractional routing in a directed network is defined in the same way, using directed paths for $\mathcal{P}_i$ and requiring the total flow on every arc $a$ to be at most $c_a$.

Linear programming duality for multicommodity flow~\cite{iri1971,onaga1971} characterizes fractional-routing feasibility by comparing the capacity cost of the network with the shortest-path cost of the demands. Assign a nonnegative length $\ell_e \ge 0$ to each edge $e \in E$, and let $d_\ell(u, v)$ denote the shortest-path distance between $u$ and $v$ in $G$ under the edge lengths $\ell = (\ell_e)_{e \in E}$. Viewing $\ell_e$ as the price per unit of capacity on edge $e$, the total capacity cost of the network is $\sum_{e \in E} c_e \ell_e$. Simultaneous fractional routing is feasible if and only if, for every choice of nonnegative edge lengths, the total capacity cost is at least the rate-weighted sum of the shortest-path distances of the source--destination pairs.

Zero-rate demands may be discarded before forming the distances in the following lemma.

\begin{lemma}[Fractional-routing duality]\label{lem:duality}
Suppose all positive-rate demand endpoints belong to one connected component of $G$. The demands $(s_i, t_i, r_i)$ in \eqref{eq:flow-demand}--\eqref{eq:flow-capacity} admit a simultaneous fractional routing if and only if, for every $\ell \in \mathbb{R}_+^E$,
\begin{equation}
\sum_{i=1}^k r_i d_\ell(s_i, t_i) \le \sum_{e \in E} c_e \ell_e. \label{eq:duality}
\end{equation}
\end{lemma}

\begin{proof}
For necessity, let $f_{iP} \ge 0$ be a feasible fractional routing satisfying \eqref{eq:flow-demand}--\eqref{eq:flow-capacity}. By the definition of shortest-path distance, every $P \in \mathcal{P}_i$ satisfies $d_\ell(s_i, t_i) \le \sum_{e \in P} \ell_e$. Using \eqref{eq:flow-demand}, interchanging the finite sums, and applying \eqref{eq:flow-capacity} gives
\begin{align}
\sum_{i=1}^k r_i d_\ell(s_i, t_i) &= \sum_{i=1}^k \sum_{P \in \mathcal{P}_i} f_{iP} d_\ell(s_i, t_i) \label{eq:duality-routing-cost-step-1} \\
&\le \sum_{i=1}^k \sum_{P \in \mathcal{P}_i} f_{iP} \sum_{e \in P} \ell_e \label{eq:duality-routing-cost-step-2} \\
&= \sum_{e \in E} \ell_e \sum_{i=1}^k \sum_{\substack{P \in \mathcal{P}_i \\ e \in P}} f_{iP} \le \sum_{e \in E} c_e \ell_e. \label{eq:duality-routing-cost}
\end{align}

For sufficiency, demands with $r_i = 0$ are served by zero flow, so we restrict attention to indices with $r_i > 0$. Each corresponding path set $\mathcal{P}_i$ is nonempty and finite because all positive-rate demand endpoints belong to one connected component of $G$. Associate a dual variable $y_i \in \mathbb{R}$ with the $i$th equality in \eqref{eq:flow-demand} and a nonnegative dual variable $\ell_e \ge 0$ with the constraint for edge $e$ in \eqref{eq:flow-capacity}. By linear programming duality, the nonnegative system \eqref{eq:flow-demand}--\eqref{eq:flow-capacity} is infeasible if and only if there exist $y = (y_i)_{i \in [k]} \in \mathbb{R}^k$ and $\ell \in \mathbb{R}_+^E$ such that, for every $i \in [k]$ with $r_i > 0$ and every $P \in \mathcal{P}_i$,
\begin{equation}
y_i \le \sum_{e \in P} \ell_e, \label{eq:farkas-path-bounds}
\end{equation}
while
\begin{equation}
\sum_{i=1}^k r_i y_i > \sum_{e \in E} c_e \ell_e. \label{eq:farkas-separation}
\end{equation}
For fixed $\ell \in \mathbb{R}_+^E$, the constraints over all paths in the nonempty set $\mathcal{P}_i$ are equivalent to
\begin{equation}
y_i \le \min_{P \in \mathcal{P}_i} \sum_{e \in P} \ell_e = d_\ell(s_i, t_i). \label{eq:farkas-distance-bound}
\end{equation}
Multiplying these upper bounds by $r_i \ge 0$ and summing over $i \in [k]$ yields
\begin{equation}
\sum_{i=1}^k r_i d_\ell(s_i, t_i) \ge \sum_{i=1}^k r_i y_i > \sum_{e \in E} c_e \ell_e, \label{eq:farkas-contradiction}
\end{equation}
contradicting \eqref{eq:duality}. Hence \eqref{eq:flow-demand}--\eqref{eq:flow-capacity} has a feasible solution.
\end{proof}

In later sections, we extend the use of shortest-path distances to pseudometrics on general finite sets. A pseudometric $d$ on a finite set $W$ is a symmetric nonnegative function satisfying $d(u, u) = 0$ and the triangle inequality; distinct points may have distance zero. The shortest-path distance $d_\ell$ on a connected graph is a finite pseudometric on its node set. Let $\binom{W}{2}$ denote the set of unordered pairs of distinct elements of $W$. For nonnegative capacities $x_{uv}$ on these pairs, write
\begin{equation}
x \cdot d = \sum_{\{u, v\} \in \binom{W}{2}} x_{uv} d(u, v). \label{eq:capacity-cost}
\end{equation}
Here $d(u, v)$ is the price per unit of capacity on the pair $\{u, v\}$, so $x \cdot d$ is the total cost of capacities $x$ at these prices.

\subsection{Network Coding}\label{subsec:multicast}
In the directed network $D = (V, \mathcal{A})$, source $s_i$ holds a message $X_i$, and the source messages are mutually independent. In multiple unicast, each destination $t_i$ needs only the message $X_i$ from its corresponding source $s_i$. In multiple multicast, every destination $t_j$ must recover all $k$ source messages $X_1, \dots, X_k$. A coding solution allows each node to encode its own source messages and previously received symbols before transmission, subject to the arc capacities. A common source rate $\rho$ is achievable if there is a sequence of coding solutions whose rates at all sources approach $\rho$ and whose decoding error probabilities tend to zero as the communication duration grows.

In an undirected network $(G,c)$, coding follows the same causal rule, with the sum of the transmission rates in the two directions of each edge bounded by its capacity. A demand collection $\mathcal{D}=((s_i,t_i,r_i))_{i\in[k]}$ is achievable by multiple-unicast coding if there is a sequence of coding solutions whose source rates approach $r_i$ and in which each prescribed destination $t_i$ recovers $X_i$ with error probability tending to zero as the communication duration grows.

Unlike multiple-unicast coding, multiple-multicast coding feasibility can be checked one destination at a time~\cite[Claim~1]{langberg2009}. To do so, construct an equivalent single-source multicast network. Add a new node $z$ holding all $k$ messages as a super-source, and add an arc $z \to s_i$ of capacity $\rho$ for every source $s_i$ to supply it with $X_i$. Delivering all messages to every destination at common rate $\rho$ in the original network is then equivalent to a single-source multicast task from $z$ at total rate $k\rho$. Li et al.~\cite{li2003linear} established the sufficiency of linear coding for single-source multicast. For networks with directed cycles, the formulation with transmission delays of Koetter and M\'edard~\cite[Theorem~15]{koetter2003algebraic} gives the same achievable rate: the minimum of the maximum flow values from the source to its destinations. Thus, this multicast task is feasible if and only if, for each destination $t_j$ considered separately, there is a feasible directed fractional flow from $z$ to $t_j$ of value $k\rho$.

The only arcs leaving $z$ are the $k$ new arcs of capacity $\rho$, so a total flow of $k\rho$ requires every new arc $z \to s_i$ to carry exactly $\rho$. Decomposing the flow into paths and deleting their first arcs $z \to s_i$ gives an equivalent condition in the original network: for every fixed destination $t_j$, there must be a feasible directed fractional flow in $D$ that simultaneously sends $\rho$ units from each of $s_1, \dots, s_k$ to $t_j$. Flows for different destinations may use the same arc because each destination is checked separately in the same network.

In particular, when all arc capacities equal $1$ and the common rate is $\rho = 1$, all capacities in the augmented network and the target flow value $k$ are integers. By the integrality theorem for maximum flows~\cite[Theorem~4.5.5]{bangjensen2009digraphs}, a directed fractional flow of value $k$ to a destination $t_j$ exists if and only if there is such an integral flow with every arc carrying either $0$ or $1$. Under unit arc capacities, decomposing this integral flow into paths and deleting the new arcs gives $k$ pairwise arc-disjoint directed paths from $s_1, \dots, s_k$ to $t_j$. This property is called strong reachability.

\begin{definition}[Strong reachability]\label{def:strong}
A unit-capacity directed network $D$ is strongly reachable if, for every $j \in [k]$, there are directed paths $P_{1j}, \dots, P_{kj}$ from $s_1, \dots, s_k$, respectively, to $t_j$, and no two of these $k$ paths share an arc.
\end{definition}

Thus, in the unit-capacity setting, strong reachability is equivalent to multiple-multicast coding at common rate one.

\section{Main Results}\label{sec:results}

Our main results are an equivalence among multiple-unicast routing, multiple-unicast coding, and single-sink routing feasibility in undirected networks and the resulting resolution of the Langberg--M\'edard conjecture.

\subsection{Coding and Routing Feasibility}
In the undirected network $(G, c)$, consider any given collection of multiple-unicast demands $\mathcal{D} = ((s_i, t_i, r_i))_{i \in [k]}$ with rates $r_i \ge 0$ and total demand rate $R = \sum_{i=1}^k r_i$. Let $T \subseteq V$ be a nonempty node subset containing all prescribed destinations $t_1, \dots, t_k$.

\begin{samepage}
\begin{theorem}[Equivalence of coding and routing feasibility]
\label{thm:main-equivalence}
Suppose the destination set $T \subseteq V$ in the undirected network $(G, c)$ satisfies the pairwise maximum-flow condition
\begin{equation}
\operatorname{maxflow}_G(u, v) \ge R \quad \text{for all distinct } u, v \in T. \label{eq:destination-connectivity}
\end{equation}
Choose any one destination $t_0 \in T$ and keep it fixed. The following are equivalent:
\begin{enumerate}
\item The given demands $\mathcal{D}$ admit a simultaneous fractional routing in $(G, c)$.
\item The given demands $\mathcal{D}$ are achievable by network coding in $(G, c)$.
\item All sources can simultaneously send their respective rates $r_i$ to the common destination $t_0$ by fractional routing in $(G, c)$.
\end{enumerate}
\end{theorem}
\end{samepage}

The pairwise maximum flows in \eqref{eq:destination-connectivity} are evaluated separately in the original network $(G, c)$. Under \eqref{eq:destination-connectivity}, single-sink routability does not depend on the choice of $t_0 \in T$, so it suffices to fix an arbitrary node in $T$, such as $t_1$. Moreover, the demand rates $r_i \ge 0$ need not be equal, the edge capacities may be arbitrary nonnegative real numbers, and source or destination locations may coincide. The proof of Theorem~\ref{thm:main-equivalence} is completed in Section~\ref{sec:reduction}.

Add a super-source $q$ to $(G, c)$ and an edge $q s_i$ of capacity $r_i$ for each $i \in [k]$, retaining parallel edges when source locations coincide. Condition~3) of Theorem~\ref{thm:main-equivalence} is equivalent to the maximum flow from $q$ to $t_0$ in this augmented network being $R$, and thus requires one maximum-flow computation. The following corollary shows that \eqref{eq:destination-connectivity} can be checked using $|T|-1$ additional pairwise maximum-flow computations in the original network, all with the same fixed destination $t_0$.

\begin{corollary}
\label{cor:maxflow-test}
Fix a destination $t_0 \in T$ in the undirected network $(G, c)$. Condition~\eqref{eq:destination-connectivity} holds if and only if
\begin{equation}
\operatorname{maxflow}_G(t_0, t) \ge R \quad \text{for all } t \in T \setminus \{t_0\}. \label{eq:root-connectivity}
\end{equation}
Consequently, at most $|T|$ maximum-flow computations suffice to verify \eqref{eq:destination-connectivity} and, when it holds, decide coding or fractional-routing feasibility of the multiple-unicast demands $\mathcal{D}$.
\end{corollary}

\begin{proof}
Every cut separating $T$ separates $t_0$ from some $t \in T \setminus \{t_0\}$, so \eqref{eq:root-connectivity} ensures that its capacity is at least $R$. The max-flow min-cut theorem then gives \eqref{eq:destination-connectivity}; the reverse implication is immediate. Combining one maximum-flow computation from $q$ to $t_0$ with Theorem~\ref{thm:main-equivalence} shows that at most $|T|$ maximum-flow computations are required in total.
\end{proof}

\subsection{The Langberg--M\'edard Conjecture}
For a directed network, applying Theorem~\ref{thm:main-equivalence} to its underlying undirected network gives the following result, which allows unequal source rates.

\begin{theorem}
\label{thm:directed-general}
Let a directed network $D$ have nonnegative real arc capacities, a nonempty destination set $T \subseteq V$, and prescribed demands $(s_i, t_i, r_i)$ for $i \in [k]$ with $t_i \in T$ and $r_i \ge 0$. Suppose that, for each destination $t \in T$ considered separately, a feasible directed fractional flow in $D$ simultaneously sends rate $r_i$ from every source $s_i$ to $t$. After replacing each arc by an undirected edge of the same capacity, the underlying undirected network $\overline{D}$ admits a simultaneous fractional routing of all prescribed demands $(s_i, t_i, r_i)$.
\end{theorem}

Taking equal source rates and using the flow characterization of multiple-multicast coding in Section~\ref{subsec:multicast} yields the Langberg--M\'edard conjecture.

\begin{corollary}[The Langberg--M\'edard conjecture]
\label{cor:lm-conjecture}
Let $D$ be a directed network with sources $s_1, \dots, s_k$ and destinations $t_1, \dots, t_k$. If a common source rate $\rho \ge 0$ is achievable by multiple-multicast coding in $D$, then $\overline{D}$ admits a simultaneous fractional routing of rate $\rho$ from $s_i$ to $t_i$ for every $i \in [k]$, with each arc replaced by an undirected edge of the same capacity. In particular, the underlying undirected graph of every strongly reachable unit-capacity $k$-pair network admits a simultaneous unit-rate fractional routing of the paired demands.
\end{corollary}

\subsection{Proof Overview}
In Theorem~\ref{thm:main-equivalence}, routing is a special case of coding, and coding feasibility implies single-sink routability by the cut-set bound, the pairwise maximum-flow condition \eqref{eq:destination-connectivity}, and the max-flow min-cut theorem. The core of the proof is the implication from single-sink routing to multiple-unicast routing: showing that single-sink routability to $t_0$ and \eqref{eq:destination-connectivity} suffice to establish the fractional-routing dual inequality \eqref{eq:duality} for the prescribed demands $(s_i, t_i, r_i)$.

Section~\ref{sec:core} first proves this metric inequality for simplified undirected networks containing only sources and destinations, with no edges between sources. Using Mader's splitting-off theorem to eliminate all sources while preserving pairwise edge connectivities among destinations, and combining the preserved connectivity with the bound in Lemma~\ref{lem:three-ball}, we obtain a cost inequality in which the added auxiliary-edge costs cancel from both sides.

Section~\ref{sec:reduction} then reduces a general undirected network to this simplified structure for each fixed rational shortest-path pseudometric via Kir\'aly's parsimonious cut-covering theorem, preserving single-sink routability and the pairwise maximum-flow condition without increasing the total capacity cost, and completes the proof of Theorem~\ref{thm:main-equivalence} by rational approximation. Finally, Section~\ref{sec:strong_reachability} proves Theorem~\ref{thm:directed-general} for directed networks by applying Theorem~\ref{thm:main-equivalence} to their underlying undirected networks, then uses the flow characterization of multiple-multicast coding from Section~\ref{subsec:multicast} to derive Corollary~\ref{cor:lm-conjecture}.

\section{Multiple-Unicast Routing When Sources Are Adjacent Only to Destinations}\label{sec:core}

This section considers an undirected network containing only distinct sources $S = \{s_1, \dots, s_k\}$ and a disjoint nonempty destination set $T = \{v_1, \dots, v_m\}$. Each source $s_i$ is adjacent only to destinations and has total incident capacity equal to its sending rate $r_i \ge 0$, with no edges between sources. Under the cut condition below, we prove that all sources can simultaneously send their respective rates $r_i$ to their prescribed destinations $t_i \in T$. For each fixed pseudometric, Section~\ref{sec:reduction} reduces the proof for a general network to this structure while preserving the cut conditions and without increasing the total capacity cost.

Collect the capacities of edges $s_i u$ in the matrix $A = (A_{iu}) \in \mathbb{R}_+^{k \times m}$, and those between distinct destinations in the vector $B = (B_{uv}) \in \mathbb{R}_+^{\binom{T}{2}}$, assigning zero to absent edges. Write $(T, B)$ for the undirected network with node set $T$ and capacity $B_{uv}$ on edge $uv$, retaining only positive-capacity edges. The $i$th row of $A$ sums to $r_i$, and the total source rate is $R = \sum_{i=1}^k r_i$. Source $s_i$ first sends $A_{iu}$ units along $s_i u$ to $u$. The remaining task is to route all this traffic simultaneously in the destination network $(T, B)$ to the respective prescribed destinations $t_i$. For $J \subseteq T$, set $A_i(J) = \sum_{u \in J} A_{iu}$.

Fix a nonempty proper subset $J$ of $T$ and consider cuts having exactly $J$ as the destinations on one side. Source $s_i$ contributes $A_i(J)$ to the cut capacity if placed on the other side, and $r_i - A_i(J)$ if placed on this side. Since sources are not adjacent, these choices can be made independently. The minimum capacity of such a cut is
\begin{equation}
B(\delta_T J) + \sum_{i=1}^k \min\{A_i(J),\, r_i - A_i(J)\}. \label{eq:core-minimum-cut}
\end{equation}
Requiring every cut separating destinations to have capacity at least $R$ is therefore equivalent to
\begin{equation}
B(\delta_T J) \ge \sum_{i=1}^k \max\{A_i(J),\, r_i - A_i(J)\} \label{eq:core-cut}
\end{equation}
for every nonempty proper subset $J$ of $T$, using $r_i - \min\{a, r_i - a\} = \max\{a, r_i - a\}$. We assume \eqref{eq:core-cut} throughout this section.

Figure~\ref{fig:destination-cut} shows the two placements of a source for a fixed destination cut. Since no edge joins two sources, each source can independently be placed on the side giving the smaller cut contribution.

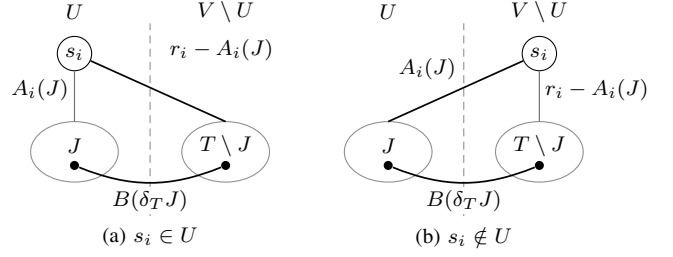
\begin{figure}[!htbp]
\centering
\begin{tikzpicture}[x=1cm,y=1cm,font=\fontsize{8}{10}\selectfont,
 terminal/.style={circle,fill=black,inner sep=1.2pt},
 source/.style={circle,draw=black,fill=white,inner sep=1.6pt},
 inside/.style={black!55,line width=.45pt},
 crossing/.style={black,line width=.65pt}]
\foreach \shift/\side in {0/0,4.15/1}{
\begin{scope}[xshift=\shift cm]
 \draw[densely dashed,black!45,line width=.4pt] (0,-.85)--(0,1.7);
 \node at (-1,1.85) {$U$};
 \node at (1,1.85) {$V\setminus U$};
 \draw[black!45,line width=.4pt] (-1,0) ellipse (.58 and .4);
 \draw[black!45,line width=.4pt] (1,0) ellipse (.58 and .4);
 \node at (-1,.08) {$J$};
 \node at (1,.08) {$T\setminus J$};
 \node[terminal] (l) at (-1,-.18) {};
 \node[terminal] (r) at (1,-.18) {};
 \draw[crossing] (l) to[bend right=22] (r);
 \node[fill=white,inner sep=1pt] at (0,-.64) {$B(\delta_T J)$};
 \ifnum\side=0
  \node[source] (s) at (-1,1.3) {$s_i$};
  \draw[inside] (s)--(-1,.4);
  \node[anchor=east,inner sep=2pt] at (-1,.82) {$A_i(J)$};
  \draw[crossing] (s)--(1,.4);
  \node[anchor=south west,inner sep=1.5pt] at (.2,1.15) {$r_i-A_i(J)$};
  \node at (0,-1.12) {(a) $s_i\in U$};
 \else
  \node[source] (s) at (1,1.3) {$s_i$};
  \draw[inside] (s)--(1,.4);
  \node[anchor=west,inner sep=2pt] at (1,.82) {$r_i-A_i(J)$};
  \draw[crossing] (s)--(-1,.4);
  \node[fill=white,inner sep=1.5pt] at (-.48,1.1) {$A_i(J)$};
  \node at (0,-1.12) {(b) $s_i\notin U$};
 \fi
\end{scope}}
\end{tikzpicture}
\caption{Two placements of a source in \eqref{eq:core-cut}. Ovals represent destination subsets, edge labels give aggregate capacities, and the vertical dashed line marks the cut. The source contributes $r_i-A_i(J)$ in (a) and $A_i(J)$ in (b). Minimizing each source contribution and adding $B(\delta_T J)$ gives the minimum capacity of a cut with this destination partition.}
\label{fig:destination-cut}
\end{figure}

\begin{samepage}
\begin{theorem}[Routing-cost inequality between destinations]
\label{thm:destination-routing}
Suppose the capacities $A$ and $B$ satisfy \eqref{eq:core-cut}. Then, for every pseudometric $d : T \times T \to \mathbb{R}_+$ on the destination set $T$ and every choice of destinations $t_1, \dots, t_k \in T$,
\begin{equation}
B \cdot d \ge \sum_{i=1}^k \sum_{u \in T} A_{iu} d(u, t_i). \label{eq:core-metric}
\end{equation}
Consequently, the demands $u \to t_i$, for $i \in [k]$ and $u \in T$, admit a simultaneous fractional routing at rates $A_{iu}$ in the destination network $(T, B)$, using only nodes in $T$.
\end{theorem}
\end{samepage}

To establish inequality \eqref{eq:core-metric} for the destination capacities $B$, we first split off the incident edges at each source to eliminate all sources, obtaining an auxiliary network on the destination set $T$ alone in which the minimum cut between any two destinations remains at least $R$. We then use this destination connectivity to lower-bound the total capacity cost of the auxiliary network and cancel the cost of the edges added by splitting from both sides, leaving \eqref{eq:core-metric} for the original destination network $(T, B)$.

\subsection{Source Elimination Preserving Connectivity}
We first consider rational $A$ and $B$. Choose an integer $N \ge 1$ such that $NA$ and $NB$ are integral, multiply all edge capacities by $2N$, and replace each edge by that many parallel unit edges; then $2N A_{iu}$ unit edges join source $s_i$ to destination $u$, and $s_i$ has $2N r_i$ incident unit edges in total. To eliminate source $s_i$, we pair up its $2N r_i$ incident unit edges into $N r_i$ pairs by splitting off: pairing a unit edge $s_i u$ with a unit edge $s_i v$ deletes both edges and adds a unit edge $uv$ between $u$ and $v$ when $u \ne v$, or simply deletes the two edges when $u = v$. Write $N H_i(u, v)$ for the number of unit-edge pairs formed at $s_i$ between $s_i u$ and $s_i v$, where $H_i(u, v) = H_i(v, u) \ge 0$ and $u = v$ is allowed. Fix a total order on $T$; the notation $\sum_{u \le v}$ counts each unordered pair in $T$ once, including pairs with $u = v$.

\begin{lemma}[Pairing weights]
\label{lem:pairing-weights}
If $A$ and $B$ are rational and satisfy \eqref{eq:core-cut}, then for each source index $i \in [k]$ there are nonnegative rational pairing weights $H_i$ with the following properties. For fixed $i$, summing over all unordered destination pairs gives
\begin{equation}
\sum_{\substack{u, v \in T \\ u \le v}} H_i(u, v) = r_i. \label{eq:H-mass}
\end{equation}
Fix also a destination $u \in T$. Summing over the remaining destinations $v \in T \setminus \{u\}$ gives
\begin{equation}
\sum_{v \in T \setminus \{u\}} H_i(u, v) + 2 H_i(u, u) = 2 A_{iu}. \label{eq:H-degree}
\end{equation}
Equation \eqref{eq:H-degree} holds for every $i \in [k]$ and every $u \in T$. It counts the total weight of occurrences of endpoint $u$ in the pairs; a pair $(u, u)$ contains two such occurrences and therefore contributes twice its weight.

For distinct destinations $u, v \in T$, define capacities
\begin{equation}
K_{uv} = B_{uv} + \frac{1}{2} \sum_{i=1}^k H_i(u, v). \label{eq:K}
\end{equation}
Write $(T, K)$ for the undirected network with node set $T$ and capacity $K_{uv}$ between distinct destinations $u, v$, retaining only positive-capacity edges. This network satisfies
\begin{equation}
K(\delta_T J) \ge R \quad (\varnothing \ne J \subsetneq T). \label{eq:Kcut}
\end{equation}
\end{lemma}

\begin{proof}
A cut with exactly $J$ as its destinations on one side has capacity at least
\begin{equation}
B(\delta_T J) + \sum_{i=1}^k \min\{A_i(J),\, r_i - A_i(J)\} \ge R, \label{eq:source-cut-bound}
\end{equation}
by \eqref{eq:core-cut}. Hence the minimum cut capacity between any two distinct destinations is at least $R$. After multiplying all capacities by $2N$ and replacing them by parallel unit edges, every edge multiplicity, and therefore every vertex degree, is even, and every cut separating two distinct destinations contains at least $2N R$ unit edges. Figure~\ref{fig:source-splitting} illustrates scaling, expansion into unit edges, and complete splitting with one source and three destinations.

\ifdefined\LMtwocolumn
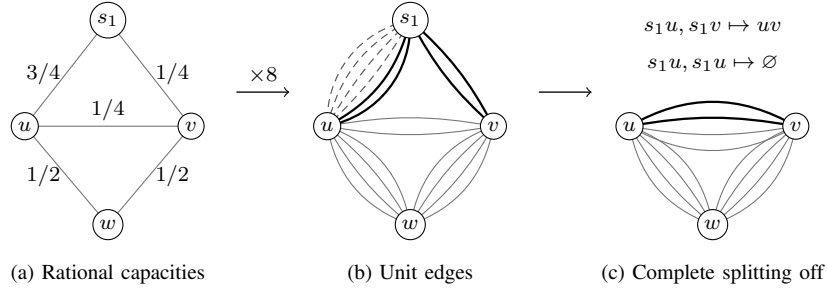
\begin{figure*}[t]
\else
\begin{figure}[!htbp]
\fi
\centering
\begin{tikzpicture}[x=1cm,y=1cm,font=\fontsize{8}{10}\selectfont,
 vertex/.style={circle,draw,fill=white,inner sep=1.5pt},
 old/.style={black!55,line width=.4pt},
 pair/.style={black,line width=.85pt},
 repeat/.style={black!65,densely dashed,line width=.45pt}]
\foreach \shift/\panel in {0/0,4/1,8/2}{
\begin{scope}[xshift=\shift cm]
 \coordinate (u) at (0,0); \coordinate (v) at (2.2,0);
 \coordinate (w) at (1.1,-1.3); \coordinate (s) at (1.1,1.4);
 \ifnum\panel=0
  \draw[old] (s)--node[left,black,inner sep=2pt] {$3/4$}(u);
  \draw[old] (s)--node[right,black,inner sep=2pt] {$1/4$}(v);
  \draw[old] (u)--node[above,black,inner sep=2pt] {$1/4$}(v);
  \draw[old] (u)--node[left,black,inner sep=2pt] {$1/2$}(w);
  \draw[old] (v)--node[right,black,inner sep=2pt] {$1/2$}(w);
  \node[vertex] at (s) {$s_1$};
  \node at (1.1,-1.95) {(a) \LMfigtext{Rational capacities}{有理容量}};
 \else
  \foreach \ang in {-30,-10,10,30}{
   \draw[old] (u) to[bend left=\ang] (w);
   \draw[old] (v) to[bend left=\ang] (w);
  }
  \ifnum\panel=1
   \foreach \ang in {-10,10}{\draw[old] (u) to[bend left=\ang] (v);}
   \foreach \ang in {-40,-24,-8,8}{\draw[repeat] (s) to[bend left=\ang] (u);}
   \foreach \ang in {24,40}{\draw[pair] (s) to[bend left=\ang] (u);}
   \foreach \ang in {-10,10}{\draw[pair] (s) to[bend left=\ang] (v);}
   \node[vertex] at (s) {$s_1$};
   \node at (1.1,-1.95) {(b) \LMfigtext{Unit edges}{单位边展开}};
  \else
   \foreach \ang in {-30,-10}{\draw[old] (u) to[bend left=\ang] (v);}
   \foreach \ang in {10,30}{\draw[pair] (u) to[bend left=\ang] (v);}
   \node[align=center] at (1.1,1.05) {$s_1u,s_1v\mapsto uv$\\[4pt]$s_1u,s_1u\mapsto\varnothing$};
   \node at (1.1,-1.95) {(c) \LMfigtext{Complete splitting off}{完整分裂}};
  \fi
 \fi
 \foreach \n in {u,v,w}{\node[vertex] at (\n) {$\n$};}
\end{scope}}
\draw[->,line width=.5pt] (2.8,.45)--node[above,inner sep=3pt] {$\times8$}(3.5,.45);
\draw[->,line width=.5pt] (6.8,.45)--(7.5,.45);
\end{tikzpicture}
\caption{Scaling and splitting at $s_1$, with $r_1=R=1$ and $N=4$. (a) Rational capacities. (b) Multiplication by $2N=8$ gives separate unit edges. (c) The two pairs $s_1u,s_1v$ drawn with thick solid lines create two $uv$ edges; the two dashed pairs $s_1u,s_1u$ are deleted. Thin destination edges remain unchanged. Dividing by $8$ gives $K_{uv}=K_{uw}=K_{vw}=1/2$ and $H_1(u,v)=H_1(u,u)=1/2$.}
\label{fig:source-splitting}
\ifdefined\LMtwocolumn
\end{figure*}
\else
\end{figure}
\fi

In a multigraph, the local edge connectivity between two nodes is the minimum number of edges in a cut separating them. Mader's splitting-off theorem (Theorem~\ref{thm:mader-external} in Appendix~\ref{sec:external-results}) permits complete splitting at an even-degree vertex with no incident bridge while preserving all pairwise local edge connectivities between the other vertices. Each split of a pair $s_i u, s_i v$ preserves even degrees, and an even-degree multigraph has no bridge because every cut has even size. Thus the hypotheses continue to hold after every split.

Completely split off the vertices $s_1, \dots, s_k$ in order. Splitting another source creates new edges only between destinations, so every neighbor of $s_i$ is a destination when $s_i$ is split. Complete splitting groups the $2N r_i$ incident unit edges of $s_i$ into $N r_i$ pairs, so
\begin{equation}
\sum_{u \le v} N H_i(u, v) = N r_i, \label{eq:scaled-pair-count}
\end{equation}
and dividing by $N$ gives \eqref{eq:H-mass}. For each destination $u \in T$, all $2N A_{iu}$ unit edges joining $s_i$ to $u$ are used exactly once during splitting: each pair with $v \ne u$ uses one edge $s_i u$, whereas each pair with $u = v$ uses two edges $s_i u$. Hence
\begin{equation}
\sum_{v \in T \setminus \{u\}} N H_i(u, v) + 2N H_i(u, u) = 2N A_{iu}, \label{eq:scaled-pair-degrees}
\end{equation}
and dividing by $N$ gives \eqref{eq:H-degree}.

For distinct destinations $u, v \in T$, the total number of unit edges between $u$ and $v$ after splitting all sources is $2N B_{uv} + \sum_{i=1}^k N H_i(u, v)$. Dividing by $2N$ gives the capacities \eqref{eq:K} of the destination network $(T, K)$, where the diagonal entries $H_i(u, u)$ do not correspond to edges. After splitting, the local edge connectivity between any two distinct destinations remains at least $2N R$. Dividing by $2N$ shows that the minimum cut capacity between any two distinct destinations in $(T, K)$ remains at least $R$. Choosing one destination on each side of any nontrivial cut separating destinations proves \eqref{eq:Kcut}.
\end{proof}

The diagonal entries $H_i(u, u)$ contribute to identity \eqref{eq:H-degree} but have zero contribution to cut capacities and to the total capacity cost. We use this convention in \eqref{eq:K} and in $H_i \cdot d$.

\subsection{From Connectivity to a Routing-Cost Lower Bound}
The preceding step produces an auxiliary destination network $(T, K)$ in which every nontrivial cut has capacity at least $R$. In the splitting at source $s_i$, each pair of destinations $u, v$ together with the target destination $t_i$ forms a triple $(u, v, t_i)$ in $T$: the sum of the target distances $d(u, t_i) + d(v, t_i)$ plus the length $d(u, v)$ of the added auxiliary edge equals the pairwise distance sum $d(u, v) + d(u, t_i) + d(v, t_i)$.

To relate this three-point distance sum to the total capacity cost $K \cdot d$ of the auxiliary network, we consider any three destinations $a, b, c \in T$ and place three pairwise disjoint open metric balls centered at $a, b, c$ with radii $\alpha_1, \alpha_2, \alpha_3$ chosen so that the radii of any two balls sum to the distance between their centers. The sum of the three radii is then
\begin{equation}
\pi_d(a, b, c) = \frac{d(a, b) + d(a, c) + d(b, c)}{2}, \label{eq:perimeter}
\end{equation}
where $a, b, c \in T$ may repeat. For each positive-radius open ball considered individually, the set of destinations inside that ball contains its own center, while the relation between the radii and the center distances excludes the other two centers; taking the destinations inside that single ball as one side of a cut and the destinations outside it as the other side therefore yields a nontrivial cut of $T$ of capacity at least $R$. Integrating these cut capacities across the radii of the three balls bounds $K \cdot d$ from below by $R \pi_d(a, b, c)$. A zero-radius open ball is empty and contributes zero to the integral. Let $M(d) = \max_{a, b, c \in T} \pi_d(a, b, c)$.

\begin{lemma}[Three-ball bound]
\label{lem:three-ball}
Suppose every nontrivial cut of the destination network $(T, K)$ has capacity at least $\lambda \ge 0$. Then, for every pseudometric $d$ on $T$,
\begin{equation}
K \cdot d \ge \lambda M(d). \label{eq:balls-bound}
\end{equation}
\end{lemma}

\begin{proof}
Fix a triple $(a_1, a_2, a_3)$. Define nonnegative radii
\begin{equation}
\alpha_p = \frac{d(a_p, a_q) + d(a_p, a_h) - d(a_q, a_h)}{2}, \label{eq:ball-radius-definition}
\end{equation}
where $\{p, q, h\} = \{1, 2, 3\}$. The triangle inequality gives $\alpha_p \ge 0$, and
\begin{equation}
\alpha_p + \alpha_q = d(a_p, a_q) \quad (p \ne q). \label{eq:radii}
\end{equation}
Set $f_p(x) = (\alpha_p - d(a_p, x))_+$, where $(z)_+ = \max\{z, 0\}$. Let $U_p = \{x \in T : f_p(x) > 0\}$. These sets are pairwise disjoint. If $f_p(x)$ and $f_q(x)$ were positive, then
\begin{equation}
d(a_p, x) + d(a_q, x) < \alpha_p + \alpha_q = d(a_p, a_q), \label{eq:ball-disjointness}
\end{equation}
contradicting the triangle inequality.

Consider the pseudometric
\begin{equation}
q(x, y) = \sum_{p=1}^3 |f_p(x) - f_p(y)|. \label{eq:ball-pseudometric}
\end{equation}
We claim that $q(x, y) \le d(x, y)$. First, replacing negative values by zero cannot increase the absolute difference between two numbers. Hence, for any $x, y \in T$,
\begin{equation}
|f_p(x) - f_p(y)| \le |d(a_p, x) - d(a_p, y)| \le d(x, y). \label{eq:ball-lipschitz}
\end{equation}
The last step follows from the triangle inequality: $d(a_p, x) \le d(a_p, y) + d(x, y)$, and the same inequality holds with $x, y$ exchanged. Thus the absolute difference of the two distances is at most $d(x, y)$. This is the 1-Lipschitz property used here.

If $x, y \in U_p$, pairwise disjointness gives $f_h(x) = f_h(y) = 0$ for every $h \ne p$, so $q(x, y) = |f_p(x) - f_p(y)| \le d(x, y)$. If exactly one point belongs to $\bigcup_p U_p$, suppose without loss of generality that $x \in U_p$ and $y$ lies outside the union. Then every $f_h(y)$ is zero, and $f_h(x) = 0$ for $h \ne p$. Again, $q(x, y) = |f_p(x) - f_p(y)| \le d(x, y)$. If both points lie outside the union, every function vanishes at both points, so $q(x, y) = 0$.

It remains to consider $x \in U_p$ and $y \in U_q$ with $p \ne q$. In this case,
\begin{align}
q(x, y) &= f_p(x) + f_q(y) \label{eq:different-ball-distance-step-1} \\
&= d(a_p, a_q) - d(a_p, x) - d(a_q, y) \le d(x, y). \label{eq:different-ball-distance}
\end{align}

For each center $a_p$ and any $0 \le z < \alpha_p$, consider the set of destinations $J_{p, z} = \{x \in T : f_p(x) > z\} = \{x \in T : d(a_p, x) < \alpha_p - z\}$ inside the ball centered at $a_p$ of radius $\alpha_p - z$. This set contains $a_p$ because $f_p(a_p) = \alpha_p > z$, whereas \eqref{eq:radii} gives $f_p(a_q) = 0 \le z$ for every $q \ne p$, so the other two centers lie in the exterior $T \setminus J_{p, z}$. Thus, the cut $\delta_T J_{p, z}$ separating the destinations $J_{p, z}$ inside this ball from the destinations $T \setminus J_{p, z}$ outside it is a nontrivial cut of $T$ and satisfies $K(\delta_T J_{p, z}) \ge \lambda$.

Figure~\ref{fig:three-balls} illustrates the expansion of the balls $J_{p, z}$ as $z$ decreases from $\alpha_p$ to zero and shows an edge crossing the boundary cut $\delta_T J_{p, z}$.

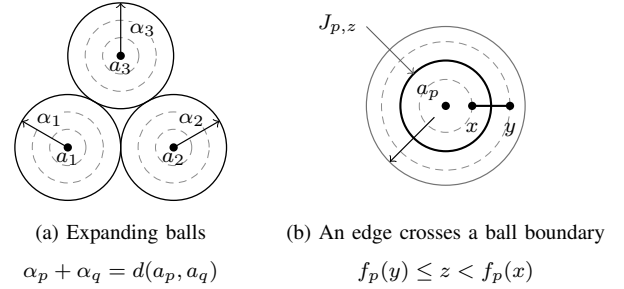
\begin{figure}[!htbp]
\centering
\begin{tikzpicture}[x=1cm,y=1cm,font=\fontsize{8}{10}\selectfont,
 dot/.style={circle,fill=black,inner sep=1.1pt}]
\begin{scope}
 \coordinate (a1) at (0,0); \coordinate (a2) at (1.4,0);
 \coordinate (a3) at (.7,1.21244);
 \foreach \n/\idx in {a1/1,a2/2,a3/3}{
  \foreach \rad in {.24,.47}{\draw[black!45,densely dashed,line width=.35pt] (\n) circle (\rad);}
  \draw[black,line width=.5pt] (\n) circle (.7);
  \node[dot] at (\n) {};
  \node[below,inner sep=2pt] at (\n) {$a_\idx$};
 }
 \draw[->,line width=.45pt] (a1)--(-.606,.35);
 \node[fill=white,inner sep=.6pt] at (-.24,.37) {$\alpha_1$};
 \draw[->,line width=.45pt] (a2)--(2.006,.35);
 \node[fill=white,inner sep=.6pt] at (1.64,.37) {$\alpha_2$};
 \draw[->,line width=.45pt] (a3)--(.7,1.91244);
 \node[fill=white,inner sep=.6pt] at (.99,1.56) {$\alpha_3$};
 \node at (.7,-1.15) {(a) \LMfigtext{Expanding balls}{球逐步扩张}};
 \node at (.7,-1.65) {$\alpha_p+\alpha_q=d(a_p,a_q)$};
\end{scope}
\begin{scope}[xshift=5cm,yshift=.55cm]
 \draw[black!55,line width=.45pt] (0,0) circle (1.05);
 \foreach \rad in {.35,.85}{\draw[black!45,densely dashed,line width=.35pt] (0,0) circle (\rad);}
 \draw[black,line width=.8pt] (0,0) circle (.6);
 \node[dot] at (0,0) {};
 \node[above left,inner sep=2pt] at (0,0) {$a_p$};
 \coordinate (x) at (.35,0);\coordinate (y) at (.85,0);
 \draw[line width=.8pt] (x)--(y);
 \node[dot,label={[label distance=1pt]below:$x$}] at (x) {};
 \node[dot,label={[label distance=1pt]below:$y$}] at (y) {};
 \node[anchor=east,inner sep=1pt] at (-1.12,1.08) {$J_{p,z}$};
 \draw[->,black!65,line width=.4pt] (-1.05,1.05)--(-.42426,.42426);
 \draw[->,line width=.45pt] (-.15,-.15)--(-.74,-.74);
 \node at (0,-1.7) {(b) \LMfigtext{An edge crosses a ball boundary}{一条边跨越球的边界}};
 \node at (0,-2.2) {$f_p(y)\leq z<f_p(x)$};
\end{scope}
\end{tikzpicture}
\caption{Ball expansion and cut counting in Lemma~\ref{lem:three-ball}. (a) As $z$ decreases to zero, the radius $\alpha_p-z$ of $J_{p,z}$ grows to $\alpha_p$; dashed contours show intermediate radii. Circular contours illustrate pseudometric balls, whose disjointness follows from \eqref{eq:radii}. (b) The thick contour contains $x$ but not $y$. The $z$ interval separating them has length $|f_p(x)-f_p(y)|$; weighting by capacities yields the integral in the proof.}
\label{fig:three-balls}
\end{figure}

For any $x, y \in T$, as $z$ ranges over $[0, \alpha_p)$, the indicator difference $|\mathbf{1}_{J_{p, z}}(x) - \mathbf{1}_{J_{p, z}}(y)|$ equals one exactly when $\min\{f_p(x), f_p(y)\} \le z < \max\{f_p(x), f_p(y)\}$, and is zero otherwise. Equivalently, for fixed $z$, this difference equals one exactly when the edge $\{x, y\}$ crosses the cut $\delta_T J_{p, z}$. This gives the integral identity
\begin{equation}
|f_p(x) - f_p(y)| = \int_0^{\alpha_p} \bigl|\mathbf{1}_{J_{p, z}}(x) - \mathbf{1}_{J_{p, z}}(y)\bigr|\,dz. \label{eq:level-set-identity}
\end{equation}
Sum this identity over all unordered pairs $\{x, y\} \in \binom{T}{2}$ with weights $K_{xy}$. Using
\begin{equation}
\sum_{\{x, y\} \in \binom{T}{2}} K_{xy} \bigl|\mathbf{1}_{J_{p, z}}(x) - \mathbf{1}_{J_{p, z}}(y)\bigr| = K(\delta_T J_{p, z}), \label{eq:level-set-cut}
\end{equation}
we obtain
\begin{align}
K \cdot d \ge K \cdot q &= \sum_{p=1}^3 \int_0^{\alpha_p} K(\delta_T J_{p, z})\,dz \label{eq:integrated-ball-bound-step-1} \\
&\ge \lambda \sum_{p=1}^3 \alpha_p = \lambda \pi_d(a_1, a_2, a_3). \label{eq:integrated-ball-bound}
\end{align}
When $\alpha_p = 0$, the corresponding integral is zero. Maximizing over $(a_1, a_2, a_3)$ gives $K \cdot d \ge \lambda M(d)$.
\end{proof}

The bound also follows from the sufficiency of the cut condition for fractional multicommodity flow with three terminals, a special case of Seymour's four-terminal theorem~\cite{seymour1980}: demands of $\lambda / 2$ between each pair of three distinct terminals satisfy the cut condition, and charging their routes by $d$ gives \eqref{eq:balls-bound}. The direct proof above exposes the metric quantity used in the next step.

\subsection{Routing Feasibility by Cost Cancellation}
We now bound the distance cost of the target demands together with the auxiliary-edge cost. The latter appears on both sides and cancels, leaving the required inequality for the original capacities $B$.

\begin{proof}[Proof of Theorem~\ref{thm:destination-routing}]
If $m = 1$, the set $T$ contains a single destination, all distances in \eqref{eq:core-metric} are zero, and all demands are local. The conclusion follows. Suppose $m \ge 2$ and first consider rational $A, B$. Take $H_i$ and $K$ from Lemma~\ref{lem:pairing-weights}, and write
\begin{equation}
Q = \sum_{i=1}^k \sum_{u \in T} A_{iu} d(u, t_i), \qquad H = \sum_{i=1}^k H_i. \label{eq:target-and-pairing-cost}
\end{equation}
Dividing both sides of identity \eqref{eq:H-degree} by two, we obtain $A_{iu} = \frac{1}{2} \sum_{v \in T \setminus \{u\}} H_i(u, v) + H_i(u, u)$. Substituting into $Q$ yields
\begin{align}
Q &= \sum_{i=1}^k \sum_{u \in T} \biggl( \frac{1}{2} \sum_{v \ne u} H_i(u, v) + H_i(u, u) \biggr) d(u, t_i) \label{eq:Qpair-step-1} \\
&= \frac{1}{2} \sum_{i=1}^k \sum_{u \le v} H_i(u, v) \bigl( d(u, t_i) + d(v, t_i) \bigr). \label{eq:Qpair}
\end{align}
The second equality combines the two terms for each distinct-endpoint pair $\{u, v\}$ with $u < v$, using $H_i(u, v) = H_i(v, u)$. For $u = v$, the term $\frac{1}{2} H_i(u, u)(d(u, t_i) + d(u, t_i))$ equals $H_i(u, u) d(u, t_i)$. Add the auxiliary-edge cost
\begin{equation}
\frac{1}{2} H \cdot d = \frac{1}{2} \sum_{i=1}^k \sum_{u \le v} H_i(u, v) d(u, v) \label{eq:auxiliary-edge-cost}
\end{equation}
to both sides of \eqref{eq:Qpair}. The half-perimeter definition \eqref{eq:perimeter} and \eqref{eq:H-mass} give
\begin{align}
Q + \frac{1}{2} H \cdot d &= \sum_{i=1}^k \sum_{u \le v} H_i(u, v) \pi_d(u, v, t_i) \label{eq:averaged-perimeter-step-1} \\
&\le M(d) \sum_{i=1}^k \sum_{u \le v} H_i(u, v) = R M(d). \label{eq:averaged-perimeter}
\end{align}
On the other hand, \eqref{eq:Kcut} and Lemma~\ref{lem:three-ball} imply
\begin{equation}
\underbrace{B \cdot d + \frac{1}{2} H \cdot d}_{K \cdot d} \ge R M(d) \ge Q + \frac{1}{2} H \cdot d. \label{eq:cancel}
\end{equation}
Subtracting $\frac{1}{2} H \cdot d$ from the first and last expressions in \eqref{eq:cancel} proves \eqref{eq:core-metric} for rational capacities.

For general nonnegative real $A, B$, choose rational sequences $0 \le A_{iu}^{(n)} \le A_{iu}$ and $B_{uv}^{(n)} \ge B_{uv}$ such that $A^{(n)} \to A$ and $B^{(n)} \to B$ as $n \to \infty$, with row sums $r_i^{(n)} = \sum_{u \in T} A_{iu}^{(n)}$. For every nonempty proper subset $J \subsetneq T$, summing $A_{iu}^{(n)} \le A_{iu}$ over $J$ and over $T \setminus J$ gives $A_i^{(n)}(J) \le A_i(J)$ and $r_i^{(n)} - A_i^{(n)}(J) \le r_i - A_i(J)$. Together with $B^{(n)}(\delta_T J) \ge B(\delta_T J)$ and \eqref{eq:core-cut}, this yields
\begin{align}
B^{(n)}(\delta_T J) &\ge B(\delta_T J) \ge \sum_{i=1}^k \max\bigl\{A_i(J),\, r_i - A_i(J)\bigr\} \label{eq:rational-cut-bound-step-1} \\
&\ge \sum_{i=1}^k \max\bigl\{A_i^{(n)}(J),\, r_i^{(n)} - A_i^{(n)}(J)\bigr\}, \label{eq:rational-cut-bound}
\end{align}
so $A^{(n)}, B^{(n)}$ also satisfy \eqref{eq:core-cut}. Applying the rational case for each $n$ and letting $n \to \infty$ gives \eqref{eq:core-metric} for real capacities.

If $R = 0$, all $A_{iu}$ are zero, so no flow is required. Finally, when $R > 0$, each term $\max\{A_i(J), r_i - A_i(J)\} \ge r_i / 2$ in \eqref{eq:core-cut}, so $B(\delta_T J) \ge R / 2 > 0$ for every nontrivial cut, and the destination network $(T, B)$ is connected. Assign any nonnegative length $\ell_e$ to each edge $e$ of $(T, B)$, and let $d_B$ be the shortest-path pseudometric on $(T, B)$. Applying \eqref{eq:core-metric} to $d_B$ and using $d_B(u, v) \le \ell_{uv}$ on each edge gives
\begin{equation}
\sum_{i=1}^k \sum_{u \in T} A_{iu} d_B(u, t_i) \le B \cdot d_B \le \sum_e B_e \ell_e. \label{eq:destination-duality}
\end{equation}
By Lemma~\ref{lem:duality}, the demands $u \to t_i$ of rates $A_{iu}$, for $i \in [k]$ and $u \in T$, admit a simultaneous fractional routing in $(T, B)$.
\end{proof}

\section{Eliminating Internal Nodes and Edges Between Sources}\label{sec:reduction}

This section completes the proof of Theorem~\ref{thm:main-equivalence}. By the max-flow min-cut theorem, both single-sink routability to the common destination $t_0$ and the pairwise maximum-flow condition \eqref{eq:destination-connectivity} between destinations can be expressed as lower-bound constraints on the capacity of every cut $U \subseteq V$. Fix a rational pseudometric $d$ on $V$, and use $d(u, v)$ as the price per unit of capacity on an edge joining $u$ and $v$. We choose new capacities $x$ on the complete graph on $V$ that continue to satisfy all of these cut-capacity lower bounds while keeping the total capacity cost at most the original cost $\sum_{e=uv \in E} c_e d(u, v)$. These capacities will be zero on all edges incident to internal nodes and on all edges between sources, with total incident capacity $r_i$ at each source $s_i$. The resulting network has the structure of Section~\ref{sec:core}, so Theorem~\ref{thm:destination-routing} can establish the routing-cost inequality needed for the original network.

\subsection{The Parsimonious Cut-Covering Theorem}
Let $h(U)$ denote the required capacity lower bound for each cut $U \subseteq V$. When the singleton lower bound is $h(\{v\}) = 0$ at each internal node $v$ and $h(\{s_i\}) = r_i$ at each source $s_i$, achieving equality $x(\delta\{v\}) = h(\{v\})$ in these singleton cut constraints forces all edges incident to internal nodes to have capacity zero and gives each source $s_i$ total incident capacity $r_i$. Thus, among the minimum-cost capacity assignments satisfying all cut constraints $x(\delta U) \ge h(U)$, we seek an optimal solution satisfying $x(\delta\{v\}) = h(\{v\})$ simultaneously at these nodes. The parsimony theorem of Kir\'aly~\cite[Theorem~2.1]{kiraly2007} provides such a guarantee. We first define the set-function properties needed to apply it.

A set function $p : 2^V \to \mathbb{R}$ is skew-supermodular if, for every $X, Y \subseteq V$, at least one of the following holds:
\begin{equation}
p(X) + p(Y) \le p(X \cap Y) + p(X \cup Y), \label{eq:skew-intersection}
\end{equation}
\begin{equation}
p(X) + p(Y) \le p(X \setminus Y) + p(Y \setminus X). \label{eq:skew-difference}
\end{equation}
For a set function $h : 2^V \to \mathbb{R}$, a node $v$ is subadditive for $h$ if
\begin{equation}
h(X) + h(\{v\}) \ge h(X \cup \{v\}) \label{eq:subadditive}
\end{equation}
for every $X \subseteq V \setminus \{v\}$. Thus, adding $v$ to $X$ increases the function value by at most the singleton value $h(\{v\})$.

The following lemma specializes the results of Kir\'aly~\cite[Claim~1.1 and Theorem~2.1]{kiraly2007} to the set functions and costs used here.

\begin{lemma}[Kir\'aly's theorem, rational form]
\label{lem:parsimonious-cover}
Let $p_1, p_2 : 2^V \to \mathbb{Q}$ be skew-supermodular. Put $h = \max\{p_1, p_2, 0\}$, taking the maximum separately for each subset of $V$. Suppose $h(U) = h(V \setminus U)$ and $h(\varnothing) = 0$. For a nonnegative rational pseudometric $d$ on $V$, consider the minimum-cost cut-covering linear program
\begin{equation}
\min\bigl\{ x \cdot d \;:\; x \ge 0,\;\; x(\delta U) \ge h(U) \text{ for every } U \subseteq V \bigr\} \label{eq:cover-lp}
\end{equation}
on the complete graph on $V$. Here $x$ assigns a nonnegative capacity to each unordered pair $\{u, v\} \in \binom{V}{2}$, and $x(\delta U)$ sums these capacities across $U$. If the feasible set is nonempty, there is an optimal solution $x$ satisfying $x(\delta\{v\}) = h(\{v\})$ simultaneously at every node $v$ that is subadditive for $h$.
\end{lemma}

\begin{proof}
When $p_1, p_2$ and the pseudometric $d$ are integer-valued, the lemma follows from Theorem~\ref{thm:kiraly-external} in Appendix~\ref{sec:external-results}.

For the rational case, choose a positive integer common denominator $L$ so that $L p_1$ and $L p_2$ are integer-valued, and a positive integer common denominator $M$ so that $M d$ is integer-valued. Setting $y = L x$, the linear program \eqref{eq:cover-lp} is equivalent to the minimum-cost cut-covering problem with integer cut bounds $L h = \max\{L p_1, L p_2, 0\}$ and integer pseudometric $M d$:
\begin{equation}
\begin{aligned}
&\min_{y \ge 0} \; y \cdot (M d) \\
&\text{subject to } y(\delta U) \ge L h(U) \quad (U \subseteq V).
\end{aligned} \label{eq:scaled-cover-lp}
\end{equation}
A node $v$ is subadditive for $h$ if and only if it is subadditive for $L h$. Applying the integer case above yields an optimal solution $y$ satisfying $y(\delta\{v\}) = L h(\{v\})$ at every such node $v$. Dividing by $L$, the rescaled vector $x = y / L$ is an optimal solution to \eqref{eq:cover-lp} satisfying $x(\delta\{v\}) = h(\{v\})$.
\end{proof}

\subsection{Constructing and Verifying the Cut Function}\label{subsec:requirements}
Let $(G, c)$ have distinct sources $S = \{s_1, \dots, s_k\}$ and a disjoint nonempty destination set $T = \{v_1, \dots, v_m\}$, with rational source rates $r_i \in \mathbb{Q}_+$ and total rate $R = \sum_{i=1}^k r_i > 0$. To apply Lemma~\ref{lem:parsimonious-cover}, we express the pairwise maximum-flow condition \eqref{eq:destination-connectivity} and single-sink routability to the fixed destination $t_0 \in T$ as cut-capacity lower bounds.

First, by the max-flow min-cut theorem, condition \eqref{eq:destination-connectivity} requires every cut separating $T$ to have capacity at least $R$. For every $U \subseteq V$, define
\begin{equation}
q_T(U) = \begin{cases}
1, & \text{if } \varnothing \ne U \cap T \ne T, \\
0, & \text{if } U \cap T = \varnothing \text{ or } T \subseteq U,
\end{cases} \label{eq:destination-cut-function}
\end{equation}
so that \eqref{eq:destination-connectivity} holds if and only if $c(\delta_G U) \ge R q_T(U)$ for every $U \subseteq V$.

Second, write $w(U) = \sum_{i : s_i \in U} r_i$ for the total rate of sources in $U$. As noted in Section~\ref{sec:results}-A, after adding a super-source $q$ joined to each $s_i$ by an edge of capacity $r_i$, the single-sink demands $(s_i, t_0, r_i)$ are simultaneously routable if and only if the maximum flow from the single source $q$ to the single sink $t_0$ in the augmented network is at least $R$. Every cut separating $q$ from $t_0$ in the augmented network has the form $\{q\} \cup U$ with $U \subseteq V \setminus \{t_0\}$; the edges crossing this cut consist of $\delta_G U$ together with the new edges $q s_i$ for $s_i \notin U$, so its capacity is $c(\delta_G U) + R - w(U)$. By the max-flow min-cut theorem, this capacity is at least $R$ if and only if $c(\delta_G U) \ge w(U)$. For a subset $U \subseteq V$ containing $t_0$, applying this condition to its complement $V \setminus U \subseteq V \setminus \{t_0\}$ gives $c(\delta_G U) = c(\delta_G(V \setminus U)) \ge w(V \setminus U) = R - w(U)$. Defining
\begin{equation}
q_0(U) = \begin{cases}
w(U), & \text{if } t_0 \notin U, \\
R - w(U), & \text{if } t_0 \in U,
\end{cases} \label{eq:single-sink-cut-function}
\end{equation}
we conclude that the single-sink demands $(s_i, t_0, r_i)$ are simultaneously routable if and only if $c(\delta_G U) \ge q_0(U)$ for every $U \subseteq V$.

Combining these two lower bounds shows that $(G, c)$ satisfies both \eqref{eq:destination-connectivity} and single-sink routability to $t_0$ if and only if $c(\delta_G U) \ge F(U)$ for every $U \subseteq V$, where
\begin{equation}
F(U) = \max\{R q_T(U),\, q_0(U),\, 0\}. \label{eq:F}
\end{equation}

\begin{lemma}
\label{lem:skew-subadditive}
Both $R q_T$ and $q_0$ are skew-supermodular, and every node $v \in V$ is subadditive for $F$.
\end{lemma}

\begin{proof}
First consider $q_0$, and take any $X, Y \subseteq V$. If $t_0$ belongs to neither $X$ nor $Y$, then $q_0$ equals the source-rate sum $w$ on each of $X, Y, X \cap Y, X \cup Y$; if $t_0$ belongs to both $X$ and $Y$, then $q_0$ equals $R - w$ on all four sets. Since $w(X) + w(Y) = w(X \cap Y) + w(X \cup Y)$, inequality \eqref{eq:skew-intersection} holds with equality in both cases. If $t_0$ belongs to exactly one of the two sets, say $t_0 \in X$ and $t_0 \notin Y$, then $t_0 \in X \setminus Y$ and $t_0 \notin Y \setminus X$. Removing $X \cap Y$ from both $X$ and $Y$ preserves the difference of their source-rate sums, namely $w(X \setminus Y) - w(Y \setminus X) = w(X) - w(Y)$, and therefore
\begin{align}
&q_0(X \setminus Y) + q_0(Y \setminus X) \notag \\
&\quad= R - w(X \setminus Y) + w(Y \setminus X) \label{eq:sink-cut-modularity-step-1} \\
&\quad= R - w(X) + w(Y) = q_0(X) + q_0(Y), \label{eq:sink-cut-modularity}
\end{align}
so \eqref{eq:skew-difference} holds with equality.

Next consider $R q_T$. It suffices to verify that the $\{0, 1\}$-valued function $q_T$ satisfies \eqref{eq:skew-intersection} or \eqref{eq:skew-difference} for any $X, Y \subseteq V$, and then multiply by $R > 0$. If $q_T(X) = q_T(Y) = 0$, the left-hand side of both inequalities is zero, so they hold by nonnegativity. If exactly one value is $1$, say $q_T(X) = 1$ and $q_T(Y) = 0$, then $Y$ does not separate $T$: when $Y \cap T = \varnothing$, the set $X \setminus Y$ contains the same destinations as $X$, so $q_T(X \setminus Y) = 1$ and \eqref{eq:skew-difference} holds; when $T \subseteq Y$, the set $X \cap Y$ contains the same destinations as $X$, so $q_T(X \cap Y) = 1$ and \eqref{eq:skew-intersection} holds. If $q_T(X) = q_T(Y) = 1$, the left-hand side equals $2$. If both $X \cap Y$ and $X \cup Y$ separate $T$, the right-hand side of \eqref{eq:skew-intersection} also equals $2$; otherwise, either $(X \cap Y) \cap T = \varnothing$ or $T \subseteq X \cup Y$. In the former case, the nonempty subsets $X \cap T$ and $Y \cap T$ are disjoint, so $(X \setminus Y) \cap T = X \cap T$ and $(Y \setminus X) \cap T = Y \cap T$ are both nonempty proper subsets of $T$; in the latter case, $(X \setminus Y) \cap T = T \setminus Y$ and $(Y \setminus X) \cap T = T \setminus X$ are likewise nonempty proper subsets of $T$. In either case, $q_T(X \setminus Y) = q_T(Y \setminus X) = 1$, so \eqref{eq:skew-difference} holds.

Finally, to verify subadditivity \eqref{eq:subadditive}, fix $v \in V$ and $X \subseteq V \setminus \{v\}$. If $v \notin S \cup T$ is an internal node, then $F(\{v\}) = 0$, and adding $v$ to $X$ changes neither its sources nor its destinations, so $F(X \cup \{v\}) = F(X)$. If $v = s_i \in S$, then $F(\{s_i\}) = r_i$, while adding $s_i$ leaves $q_T$ unchanged and increases $q_0$ by at most $r_i$, giving $F(X \cup \{s_i\}) \le F(X) + r_i$. If $v \in T$, then $F(\{v\}) = R$, and since $F$ never exceeds $R$, we have $F(X \cup \{v\}) \le R \le F(X) + F(\{v\})$.
\end{proof}

\subsection{Capacity Transformation Without Cost Increase}

\begin{lemma}[Capacity transformation without increasing total capacity cost]
\label{lem:capacity-transformation}
Suppose $c(\delta_G U) \ge F(U)$ for every $U \subseteq V$. For every rational pseudometric $d$ on $V$, there are a nonnegative rational capacity matrix $A \in \mathbb{Q}_+^{k \times m}$ and a nonnegative rational capacity vector $B \in \mathbb{Q}_+^{\binom{T}{2}}$. Here $A_{iu}$ is the new capacity between source $s_i$ and destination $u$, and $B_{uv}$ is the new capacity between distinct destinations $u, v \in T$. The $i$th row of $A$ sums to $r_i$, the capacities $A$ and $B$ satisfy \eqref{eq:core-cut}, and
\begin{equation}
\sum_{i=1}^k \sum_{u \in T} A_{iu} d(s_i, u) + B \cdot (d|_T) \le \sum_{e=uv \in E} c_e d(u, v). \label{eq:cost-reduction}
\end{equation}
Here and below, $d|_T$ denotes the restriction of $d$ to the destination set $T$.
\end{lemma}

\begin{proof}
Fix a rational pseudometric $d$ on $V$. On the complete graph on $V$, consider the minimum-cost cut-covering linear program
\begin{equation}
\min_{x \in \mathbb{R}_+^{\binom{V}{2}}} \; x \cdot d \quad \text{subject to} \quad x(\delta U) \ge F(U) \quad (U \subseteq V), \label{eq:F-cover-lp}
\end{equation}
which is \eqref{eq:cover-lp} in Lemma~\ref{lem:parsimonious-cover} with $h = F$.

Combine parallel-edge capacities in the original network $(G, c)$ and assign zero capacity to each absent pair, giving a capacity vector $c^{\mathrm{ext}}$ on the complete graph. The total capacity crossing each $U \subseteq V$ is unchanged, so
\begin{equation}
c^{\mathrm{ext}}(\delta U) = c(\delta_G U) \ge F(U). \label{eq:extended-capacity-cut}
\end{equation}
Its total capacity cost is $c^{\mathrm{ext}} \cdot d = \sum_{e=uv \in E} c_e d(u, v)$. Thus $c^{\mathrm{ext}}$ is a feasible solution to this linear program.

By \eqref{eq:F} and Lemma~\ref{lem:skew-subadditive}, the cut function $F$ satisfies every hypothesis of Lemma~\ref{lem:parsimonious-cover}, and every node is subadditive for $F$. Thus this linear program has an optimal solution $x$ in which every singleton cut constraint holds with equality. Write $\deg_x(v) = x(\delta\{v\})$ for the total capacity incident to node $v$. Then
\begin{equation}
\deg_x(v) = F(\{v\}) = \begin{cases}
r_i, & v = s_i \in S, \\
R, & v \in T, \\
0, & v \notin S \cup T.
\end{cases} \label{eq:degrees}
\end{equation}
Adding \eqref{eq:degrees} as equality constraints retains the optimal solution above, so the optimal value is unchanged. All constraint coefficients, right-hand sides, and objective coefficients $d(u, v)$ are rational, and each edge capacity satisfies $0 \le x_{uv} \le \deg_x(u) \le R$. This rational linear program has a nonempty bounded feasible set and hence a rational optimal solution. Choose such a solution $x$.

We identify the edges of capacity zero in two steps.
\begin{enumerate}
\item \emph{Eliminate internal nodes.} By \eqref{eq:degrees}, each internal node $v \notin S \cup T$ has total incident capacity $\deg_x(v) = 0$. Since edge capacities are nonnegative, every edge incident to an internal node has capacity zero.
\item \emph{Eliminate edges between sources.} Let $E(S)$ be the edges of the complete graph with both endpoints in the source set $S$. Summing $\deg_x(s_i) = r_i$ over all $k$ sources counts each edge in $E(S)$ twice, since both endpoints lie in $S$, and each edge crossing $\delta S$ once. Together with the cut constraint $x(\delta S) \ge q_0(S) = w(S) = R$, this gives
\begin{align}
R &= \sum_{i=1}^k \deg_x(s_i) \label{eq:source-degree-sum-step-1} \\
&= 2 x(E(S)) + x(\delta S) \ge 2 x(E(S)) + R. \label{eq:source-degree-sum}
\end{align}
It follows that $x(E(S)) = 0$, so no edge between sources has positive capacity.
\end{enumerate}

Thus $x$ has positive capacity only between sources and destinations or between distinct destinations. Set $A_{iu} = x_{s_i u}$ for $i \in [k],\, u \in T$, and $B_{uv} = x_{uv}$ for $\{u, v\} \in \binom{T}{2}$. Each source $s_i$ is joined by positive-capacity edges only to destinations and has $\deg_x(s_i) = r_i$, so the $i$th row of $A$ sums to $r_i$.

For any nonempty proper subset $J$ of $T$, every cut of the complete graph with $U \cap T = J$ separates destinations and hence satisfies $x(\delta U) \ge F(U) = R$. As in the calculation preceding \eqref{eq:core-cut} in Section~\ref{sec:core}, choose the side of each source independently to minimize the cut capacity. This gives
\begin{equation}
B(\delta_T J) + \sum_{i=1}^k \min\{A_i(J),\, r_i - A_i(J)\} \ge R. \label{eq:reduced-cut-bound}
\end{equation}
Rearranging and using $r_i - \min\{a, r_i - a\} = \max\{a, r_i - a\}$ shows that $A$ and $B$ satisfy \eqref{eq:core-cut}. Finally, optimality of $x$ and feasibility of $c^{\mathrm{ext}}$ give
\begin{align}
&\sum_{i=1}^k \sum_{u \in T} A_{iu} d(s_i, u) + B \cdot (d|_T) \label{eq:reduced-capacity-cost-step-1} \\
&\quad= x \cdot d \le c^{\mathrm{ext}} \cdot d = \sum_{e=uv \in E} c_e d(u, v), \label{eq:reduced-capacity-cost}
\end{align}
which proves \eqref{eq:cost-reduction}.
\end{proof}

\subsection{Completing the Proof of Theorem~\ref{thm:main-equivalence}}

\begin{proof}[Proof of Theorem~\ref{thm:main-equivalence}]
If $R = 0$, all three conditions hold without transmitting any messages. Assume $R > 0$ and discard demands with $r_i = 0$. Condition~1) implies Condition~2), because forwarding messages along fractional routes is a coding solution; real path-flow rates are attained by limits of finite-block routing schemes.

Next, we show that Condition~2) implies Condition~3). Suppose the demands $(s_i,t_i,r_i)$ are achievable by multiple-unicast coding in $(G,c)$, and write $w(U)=\sum_{i:s_i\in U}r_i$. By the characterization in Section~\ref{subsec:requirements}, it suffices to verify $c(\delta_G U)\ge w(U)$ for every $U\subseteq V\setminus\{t_0\}$. If $U\cap T=\varnothing$, all destinations are outside $U$ and collectively recover every message originating inside $U$. These messages are mutually independent and independent of the source messages outside $U$. The cut-set bound~\cite{harvey2006} therefore gives $c(\delta_G U)\ge w(U)$, since the total rate of information transmitted from $U$ to its complement cannot exceed the cut capacity. The inequality also holds for limits of achievable rates with vanishing decoding error. If $U\cap T\ne\varnothing$, choose $t\in U\cap T$; the cut separates $t$ from $t_0$, and \eqref{eq:destination-connectivity} gives $c(\delta_G U)\ge\operatorname{maxflow}_G(t,t_0)\ge R\ge w(U)$. Thus Condition~3) holds.

Finally, we show that Condition~3) implies Condition~1). If $|T|=1$, all prescribed destinations coincide with $t_0$, so the two routing tasks are identical. Assume $|T|\ge2$, that $(s_i,t_0,r_i)$ are simultaneously routable, and that \eqref{eq:destination-connectivity} holds. After deleting zero-capacity edges, condition \eqref{eq:destination-connectivity} with $R > 0$ places all destinations of $T$ in one connected component, and single-sink routability to $t_0$ places all sources in that same component; hence we may assume that $G$ is connected. Moreover, attaching a new leaf $s_i^+$ to each $s_i$ by an edge of capacity $r_i$ and moving the $i$th source to $s_i^+$ preserves pairwise maximum flows in $T$, single-sink routability to $t_0$, and multiple-unicast routability, so we may assume that $s_1, \dots, s_k$ are distinct and $S \cap T = \varnothing$.

By Lemma~\ref{lem:duality}, it suffices to prove that for every $\ell \in \mathbb{R}_+^E$, the shortest-path pseudometric $d = d_\ell$ satisfies \eqref{eq:duality}. First suppose $r_i \in \mathbb{Q}_+$ and $\ell_e \in \mathbb{Q}_+$, so that $d = d_\ell$ takes finite rational values. By Section~\ref{subsec:requirements}, Condition~3) and \eqref{eq:destination-connectivity} are equivalent to $c(\delta_G U) \ge F(U)$ for every $U \subseteq V$. Applying Lemma~\ref{lem:capacity-transformation} yields rational capacities $A$ and $B$, and applying Theorem~\ref{thm:destination-routing} to $d|_T$ yields \eqref{eq:core-metric}. Using the row sums $\sum_{u \in T} A_{iu} = r_i$, the triangle inequality $d(s_i, t_i) \le d(s_i, u) + d(u, t_i)$, inequality \eqref{eq:core-metric}, and inequality \eqref{eq:cost-reduction} in turn gives
\begin{align}
\sum_{i=1}^k r_i d(s_i, t_i) &= \sum_{i=1}^k \sum_{u \in T} A_{iu} d(s_i, t_i) \label{eq:final-duality-step-1} \\
&\le \sum_{i=1}^k \sum_{u \in T} A_{iu} \bigl( d(s_i, u) + d(u, t_i) \bigr) \label{eq:final-duality-step-2} \\
&\le \sum_{i=1}^k \sum_{u \in T} A_{iu} d(s_i, u) + B \cdot (d|_T) \label{eq:final-duality-step-3} \\
&\le \sum_{e=uv \in E} c_e d(u, v) \le \sum_{e \in E} c_e \ell_e, \label{eq:final-duality}
\end{align}
where the last step uses $d(u, v) \le \ell_e$ on each edge $e = uv \in E$.

For general real $\ell \in \mathbb{R}_+^E$ and real rates $r_i > 0$, choose rational nonnegative approximations $\ell^{(n)} \to \ell$ and rational rates $0 < r_i^{(n)} \le r_i$ with $r_i^{(n)} \to r_i$. Scaling down the flow of source $s_i$ by $r_i^{(n)} / r_i \le 1$ in the single-sink routing to $t_0$ shows that $(s_i, t_0, r_i^{(n)})$ remain simultaneously routable, while $\operatorname{maxflow}_G(u, v) \ge R \ge \sum_{i=1}^k r_i^{(n)}$ for distinct $u, v \in T$. Thus \eqref{eq:final-duality} holds for each $n$; by the continuity of shortest-path distances in the edge lengths, letting $n \to \infty$ gives \eqref{eq:duality}, and Lemma~\ref{lem:duality} yields Condition~1).
\end{proof}

\section{The Langberg--M\'edard Conjecture}\label{sec:strong_reachability}

We first apply Theorem~\ref{thm:main-equivalence} to the underlying undirected network of $D$ to prove Theorem~\ref{thm:directed-general}, then use the flow characterization of multiple-multicast coding in Section~\ref{subsec:multicast} to derive Corollary~\ref{cor:lm-conjecture}. For any subset $U \subseteq V$, let $\delta_D^+(U)$ and $\delta_D^-(U)$ denote the sets of outgoing and incoming arcs of $U$ in the directed network $D$, respectively.

\begin{proof}[Proof of Theorem~\ref{thm:directed-general}]
Put $R = \sum_{i=1}^k r_i$. If $R = 0$, no flow is required; assume $R > 0$. Fix any destination $t_0 \in T$. By hypothesis, there is a feasible directed fractional flow in $D$ that simultaneously sends rate $r_i$ from each source $s_i$ to $t_0$. Removing the arc directions preserves feasibility of this flow, so the single-sink demands $(s_i, t_0, r_i)$ admit a simultaneous fractional routing in $\overline{D}$.

We next verify that $\overline{D}$ satisfies the pairwise destination max-flow condition \eqref{eq:destination-connectivity}. Consider any cut $U \subseteq V$ separating $T$, put $R_U = \sum_{i : s_i \in U} r_i$, and choose destinations $t_a \in T \setminus U$ and $t_b \in U \cap T$. The feasible directed flow to $t_a \notin U$ must carry all traffic from sources inside $U$ across outgoing arcs, giving $c(\delta_D^+(U)) \ge R_U$. Similarly, the feasible directed flow to $t_b \in U$ must carry all traffic from sources outside $U$ across incoming arcs, giving $c(\delta_D^-(U)) \ge R - R_U$. The outgoing arc set $\delta_D^+(U)$ and the incoming arc set $\delta_D^-(U)$ are disjoint. Therefore, the cut capacity in the underlying undirected network $\overline{D}$ satisfies
\begin{equation}
c(\delta_{\overline{D}} U) = c(\delta_D^+(U)) + c(\delta_D^-(U)) \ge R_U + (R - R_U) = R. \label{eq:directed-cut-bound}
\end{equation}
By the max-flow min-cut theorem, the maximum flow in $\overline{D}$ between every pair of distinct destinations in $T$ is at least $R$. Together with the single-sink routability to $t_0$, Theorem~\ref{thm:main-equivalence} implies that the prescribed demands $(s_i, t_i, r_i)$ admit a simultaneous fractional routing in $\overline{D}$.
\end{proof}

\begin{proof}[Proof of Corollary~\ref{cor:lm-conjecture}]
Suppose a common source rate $\rho \ge 0$ is achievable by multiple-multicast coding in $D$. By the flow characterization of multiple-multicast coding in Section~\ref{subsec:multicast}, for each destination $t_j$ ($j \in [k]$) considered separately, there is a feasible directed fractional flow in $D$ simultaneously sending $\rho$ units of flow from every source $s_1, \dots, s_k$ to $t_j$. Applying Theorem~\ref{thm:directed-general} with $T = \{t_1, \dots, t_k\}$ and $r_i = \rho$ for $i \in [k]$ yields a simultaneous fractional routing of rate $\rho$ from each $s_i$ to $t_i$ in $\overline{D}$. In particular, when all arc capacities equal one and $D$ is strongly reachable, assigning one unit of flow to each arc-disjoint path for each destination satisfies the hypothesis of Theorem~\ref{thm:directed-general} with $r_i = 1$, proving the unit-rate statement.
\end{proof}

\section{Conclusion}\label{sec:conclusion}

We proved an equivalence among multiple-unicast routing, multiple-unicast coding, and single-sink routing feasibility in undirected networks, and used it to prove the Langberg--M\'edard conjecture. Although coding can support rates unattainable by fractional routing in general undirected networks, whenever the maximum flow between every pair of distinct destinations is at least the total demand rate, coding and fractional routing for the prescribed demands are feasible exactly when all source rates can be routed simultaneously to one fixed destination. Thus, coding and fractional routing have the same feasible rate vectors within the range defined by this maximum-flow condition. Verifying the condition and, when it holds, deciding coding or routing feasibility require at most as many maximum-flow computations as there are destinations.

The proof successively eliminated internal nodes, edges between sources, and then sources, reducing the routing-cost analysis to an auxiliary network containing only destinations. The resulting cost inequality carried over to the original network and yielded routing feasibility through multicommodity-flow duality. This reduction method not only applies to other multicommodity-flow and network information-flow problems, but also identifies a sufficient condition for equivalence between coding and routing feasibility in undirected networks in terms of connectivity among destinations. Determining the exact gap between coding and routing rates under weaker connectivity conditions is an important direction for future research.

\appendices
\section{Statements of the Results Used}\label{sec:external-results}
\setcounter{equation}{0}
\renewcommand{\theequation}{\Alph{section}.\arabic{equation}}
\renewcommand{\theHequation}{appendix.\Alph{section}.\arabic{equation}}

For ease of reference, this appendix states the forms of two existing results used in the proofs.

\subsection{Mader's Splitting-Off Theorem}

Let $H$ be a finite undirected multigraph without loops; parallel edges are allowed. For distinct vertices $u,v$, let $\lambda_H(u,v)$ be the minimum number of edges in a cut separating $u$ and $v$. Splitting a pair of edges $xu,xv$ at $x$ deletes both edges and, if $u\ne v$, adds an edge $uv$; if $u=v$, the two edges are simply deleted.

The following complete-splitting form of Mader's theorem~\cite{mader1978} follows from the formulation in Lau and Yung~\cite[Theorem~1.2 and Section~1.1]{lau2013}.

\begin{externaltheorem}[Mader's splitting-off theorem]\label{thm:mader-external}
If a vertex $x$ has even degree and no incident bridge, then all edges incident to $x$ can be paired and split off in a sequence that preserves the local edge connectivity between every pair of distinct vertices other than $x$. After deleting the resulting isolated vertex $x$, the multigraph $H'$ on $V(H)\setminus\{x\}$ satisfies
\begin{equation}
\lambda_{H'}(u,v)=\lambda_H(u,v)
\quad\text{for all distinct }u,v\ne x. \label{eq:app-connectivity-preservation}
\end{equation}
\end{externaltheorem}

Lemma~\ref{lem:pairing-weights} applies this result after multiplying the rational capacities by $2N$ and replacing them by parallel unit edges. Every vertex then has even degree, so every cut has even size and there is no bridge.

\subsection{Kir\'aly's Parsimonious Cut-Covering Theorem}

For a finite set $V$, a function $p:2^V\to\mathbb{Z}$ is skew-supermodular if, for every $X,Y\subseteq V$, at least one of
\begin{align}
p(X)+p(Y)&\le p(X\cap Y)+p(X\cup Y), \label{eq:app-skew-intersection} \\
p(X)+p(Y)&\le p(X\setminus Y)+p(Y\setminus X) \label{eq:app-skew-difference}
\end{align}
holds. A vertex $v$ is subadditive for a function $h$ if
\begin{equation}
h(X\cup\{v\})\le h(X)+h(\{v\})
\quad\text{for every }X\subseteq V\setminus\{v\}. \label{eq:app-subadditivity}
\end{equation}
The following integer-data form used in Lemma~\ref{lem:parsimonious-cover} combines Claim~1.1 and Theorem~2.1 of Kir\'aly~\cite{kiraly2007}.

\begin{externaltheorem}[Kir\'aly's parsimonious cut-covering theorem]\label{thm:kiraly-external}
Let $p_1,p_2:2^V\to\mathbb{Z}$ be skew-supermodular, and define $h(U)=\max\{p_1(U),p_2(U),0\}$. Suppose $h(U)=h(V\setminus U)$ for every $U\subseteq V$ and $h(\varnothing)=0$. Let $d$ be a nonnegative integer-valued pseudometric on $V$. On the complete graph with edge set $E_V=\binom{V}{2}$, consider
\begin{equation}
\begin{aligned}
&\min_{x \in \mathbb{R}_+^{E_V}} \; \sum_{\{u,v\}\in E_V}x_{uv}d(u,v) \\
&\text{subject to } x(\delta U)\ge h(U) \quad (U\subseteq V).
\end{aligned} \label{eq:app-cover-lp}
\end{equation}
Here $x(\delta U)$ is the sum of capacities on pairs with exactly one endpoint in $U$. If the feasible set is nonempty, there is an optimal solution $x$ satisfying
\begin{equation}
x(\delta\{v\})=h(\{v\}) \label{eq:app-parsimony}
\end{equation}
simultaneously at every vertex $v$ that is subadditive for $h$.
\end{externaltheorem}

The proof of Lemma~\ref{lem:parsimonious-cover} extends this statement to rational cut bounds and rational pseudometrics by clearing denominators and rescaling the capacities.

\section*{Acknowledgment}\label{sec:acknowledgment}
During the research and preparation of this manuscript, the authors used OpenAI's ChatGPT and Codex to assist with proof brainstorming, initial drafting, manuscript revision, and figure generation, and used Google's Gemini for proof verification and exposition refinement. The authors have thoroughly verified all mathematical derivations and text, and take full responsibility for the accuracy, originality, and integrity of the paper.

\bibliographystyle{IEEEtran}
\bibliography{references}
\end{document}